\documentclass[11pt, a4paper]{article}

\usepackage[utf8]{inputenc}    
\usepackage[T1]{fontenc}       
\usepackage[english]{babel}    
\usepackage{amsmath, amssymb, amsthm, mathtools} 
\usepackage{bm}                
\usepackage{textcomp}          
\usepackage{enumitem}          
\usepackage{geometry}          
\usepackage{graphicx}          
\usepackage{tikz-cd}		   
\usepackage{xcolor}            
\DeclareGraphicsExtensions{.pdf,.eps,.jpg,.png}
\graphicspath{{figures/}{figure/}{pictures/}{picture/}{pic/}{pics/}{image/}{images/}}

\usepackage[colorlinks=true, allcolors=blue]{hyperref} 
\usepackage[capitalize]{cleveref}
\usepackage[numbers, sort&compress]{natbib}

\makeatletter

\newcommand{\Rmnum}[1]{\expandafter\@slowromancap\romannumeral #1@}
\makeatother

\newcommand{\E}{\mathbb{E}}
\newcommand{\Pp}{\mathbb{P}}
\newcommand{\Var}{\operatorname{Var}}
\newcommand{\cE}{\mathcal{E}}
\newcommand{\free}{\mathrm{free}}
\newcommand{\surf}{\mathrm{surf}}
\newcommand{\per}{\mathrm{per}}

\theoremstyle{plain}
\newtheorem{theorem}{Theorem}
\newtheorem{lemma}[theorem]{Lemma}     
\newtheorem{proposition}[theorem]{Proposition}
\newtheorem{corollary}[theorem]{Corollary}

\theoremstyle{definition}

\theoremstyle{remark}
\newtheorem*{remark}{Remark}

\renewcommand{\[}{\begin{equation}}
\renewcommand{\]}{\end{equation}}

\title{\Large\bf Subextensive Random Boundary Perturbations in the Short-Range Edwards--Anderson Model}
\author{Hexiang Wang \and Keheng Zhu \and Mauris Chueng}
\newcommand{\Addresses}{{
		\bigskip
		\footnotesize
		
		\textsc{Hexiang Wang}, \textsc{School of Mathematical Sciences, Nankai University, Tianjin, 300071, China}\par\nopagebreak
		\texttt{Kui6539@outlook.com}
		\medskip
		
		\textsc{Keheng Zhu}, \textsc{Academy for Multidisciplinary Studies, School of Mathematics Sciences, Capital Normal
			University, Beijing, 100048, China}\par\nopagebreak
		\texttt{hexistartop@gmail.com}
		\medskip
		
		\textsc{Mauris Chueng}, \textsc{School of Statistics and Data Science, Jilin University of Finance and Economics, Changchun, 130117, China}\par\nopagebreak
		\texttt{maurischueng@gmail.com}
		\medskip
}}
\date{}

\begin{document}

\maketitle
\begin{abstract}
We consider the nearest-neighbor Edwards--Anderson Ising model on cubic boxes with random perturbations supported at the boundary.  We prove that any perturbation admitting an energy envelope that is negligible compared with the volume, both in expectation and almost surely, leaves the limiting quenched specific free energy unchanged.  For independent finite-variance bulk disorder in dimension at least two, an Efron--Stein estimate also yields almost-sure self-averaging along the full sequence of boxes.  The hypotheses are verified for i.i.d. scalar surface fields, fixed random exterior spins, and periodic wrap-around bonds, with an $ O(L^{-1}) $ comparison of quenched means.  The result concerns the specific free energy and does not assert convergence of finite-volume Gibbs measures.
\end{abstract}

\section{Introduction}

The thermodynamic limit of the quenched pressure for finite-dimensional disordered systems is well established; see, for example, \cite{GoulartRosa1982,ContucciGiardinaPule2004,ContucciStarr2009}.  For a finite-range model, changing the boundary condition modifies only a surface-order family of interactions.  The resulting independence of the limiting specific free energy is therefore expected, but it is useful to record a precise pathwise statement that also gives almost-sure self-averaging under minimal moment assumptions.

Random boundary conditions have a second, substantially more delicate role.  Even when the specific free energy converges, the associated finite-volume Gibbs measures may fail to converge and may exhibit chaotic size dependence~\cite{NewmanStein1992}.  This distinction and the associated metastate formalism are discussed in \cite{AizenmanWehr1990,NewmanStein1997,EndoVanEnterLeNy2019}; rigorous low-temperature results for the ferromagnetic Ising model with random exterior spins appear in \cite{VanEnterNetocnySchaap2005}.  The present note does not address that state-selection problem.  Its purpose is to isolate the boundary-stability argument for the specific free energy and to distinguish a random scalar surface field from a literal exterior-spin boundary condition.

\section{Model and main results}

Fix $ d\geq2 $, $ \beta>0 $, and
\[
\Lambda_L=\{1,\ldots,L\}^d,
\qquad V_L=|\Lambda_L|=L^d.
\]
Let $ \cE_L $ be the set of unoriented internal nearest-neighbor edges.  For $ L\geq2 $,
\begin{equation}
	\label{eq:geometry}
|\cE_L|=d(L-1)L^{d-1},
\qquad
|\partial\Lambda_L|=L^d-(L-2)^d\leq 2dL^{d-1}.
\end{equation}
The spin space is $ \{-1,+1\}^{\Lambda_L} $.  Let $ (J_e) $ be i.i.d. couplings satisfying
\begin{equation}
	\label{eq:bulk}
\E J_e=0,
\qquad
\Var(J_e)=v_J<\infty.
\end{equation}
The bulk field is either deterministic, $ g_x=h\in\mathbb{R} $, or an i.i.d. field independent of $ J $ with $ \Var(g_x)=v_g<\infty $.  In the deterministic case we set $ v_g=0 $.  The free-boundary Hamiltonian is
\begin{equation}
	\label{eq:hfree}
H_L^{\free}(\sigma)
=-\sum_{e=\langle x,y\rangle\in\cE_L}J_e\sigma_x\sigma_y
-\sum_{x\in\Lambda_L}g_x\sigma_x.
\end{equation}

For a boundary perturbation $ B_L $, put
\begin{equation}
	\label{eq:hb}
H_L^B(\sigma)=H_L^{\free}(\sigma)+B_L(\sigma).
\end{equation}

\begin{figure}[htbp]
	\label{fig:ea_random_boundary_structure}
	\centering
	\includegraphics[width=0.9\textwidth]{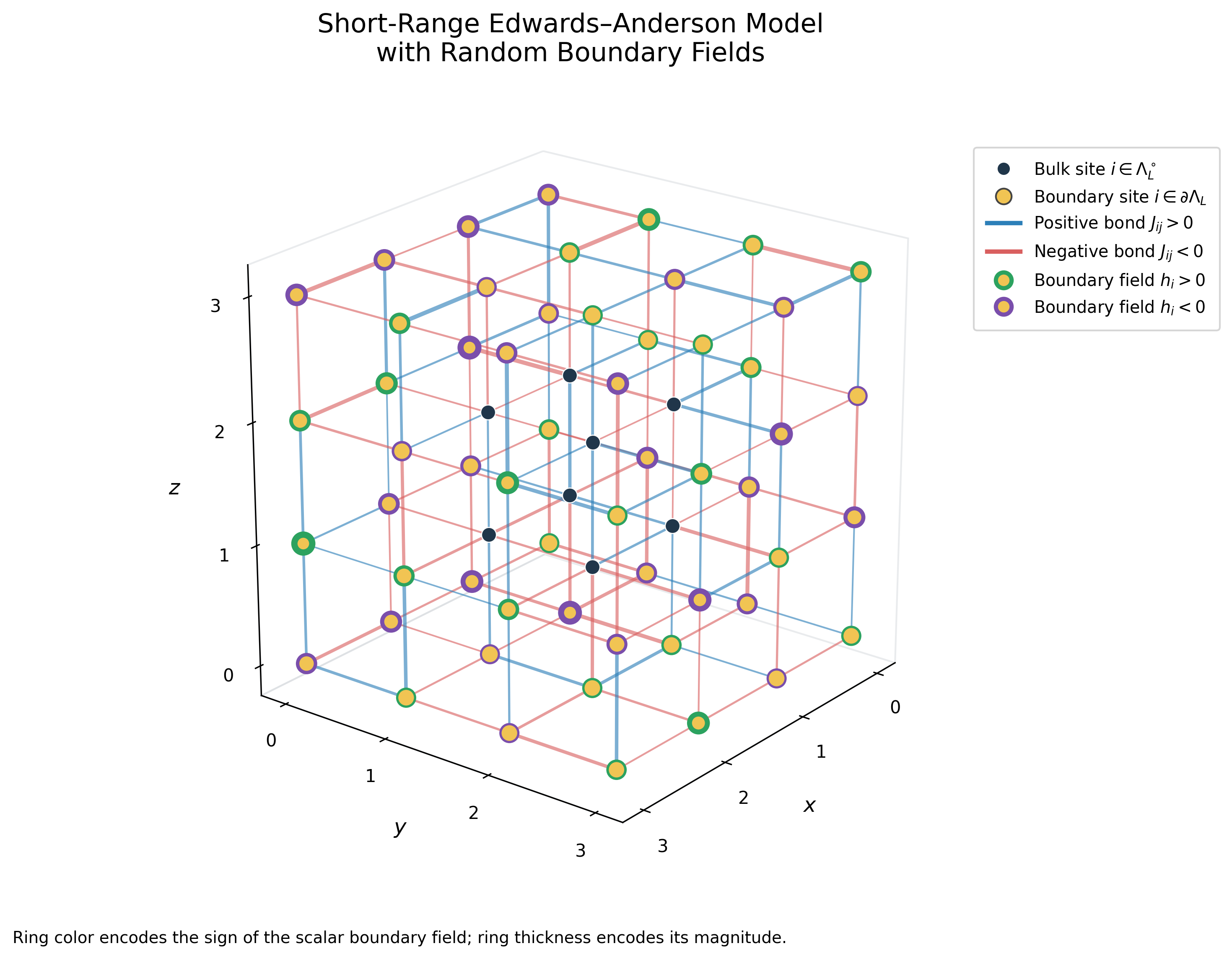}
	\caption{Geometric realization of the 3D Edwards-Anderson model under subextensive random boundary perturbations. The interior nodes represent bulk spins governed by random couplings $J_e$, while the highlighted surface layer $\partial\Lambda_L$ accommodates independent boundary fields or exterior configuration forces $B_L(\sigma)$, visually illustrating the partition defined in \eqref{eq:geometry} and \eqref{eq:hb}.}
\end{figure}

The spatial layout and the interaction boundaries are schematically illustrated in Figure~\ref{fig:ea_random_boundary_structure}. Whenever $ B_L $ is random, we assume that it is measurable and that the
resulting sample specific free energy is integrable.

Whenever $ B_L $ is random, we assume that it is measurable and that the
resulting sample specific free energy is integrable.
For $ a\in\{\free,B\} $, define
\begin{equation}
	\label{eq:freeenergy}
Z_L^a=\sum_{\sigma}\exp\bigl(-\beta H_L^a(\sigma)\bigr),
\qquad
\widehat f_L^a=-\frac{1}{\beta V_L}\log Z_L^a,
\qquad
f_L^a=\E\widehat f_L^a.
\end{equation}
All variables for all $ L $ are defined on one probability space whenever an almost-sure limit along the entire sequence is asserted.

\begin{theorem}[Subextensive boundary perturbations]\label{thm:main}
Suppose there are integrable random variables $ A_L\geq0 $ such that
\begin{equation}
	\label{eq:envelope}
\sup_{\sigma}|B_L(\sigma)|\leq A_L,
\qquad
\frac{\E A_L}{V_L}\longrightarrow0,
\qquad
\frac{A_L}{V_L}\longrightarrow0
\quad\Pp\text{-almost surely}.
\end{equation}
Then there is a deterministic constant $ f_\infty $ such that
\begin{equation}
	\label{eq:main-conclusions}
f_L^B\longrightarrow f_\infty,
\qquad
\widehat f_L^B-f_L^B\stackrel{a.s.}{\longrightarrow}0,
\end{equation}
and consequently
\begin{equation}
	\label{eq:main-as-limit}
\widehat f_L^B\stackrel{a.s.}{\longrightarrow} f_\infty.
\end{equation}
The constant $ f_\infty $ is the free-boundary thermodynamic limit.
\end{theorem}

The following two models verify the envelope hypothesis.

\begin{corollary}[Random scalar surface field]\label{cor:surface}
Let $ (\eta_x)_{x\in\mathbb Z^d} $ be i.i.d. with $ \E\eta_0^2<\infty $, and define
\begin{equation}
H_L^{\surf}(\sigma)=H_L^{\free}(\sigma)
-\sum_{x\in\partial\Lambda_L}\eta_x\sigma_x.
\label{eq:surface-hamiltonian}
\end{equation}
Then the conclusions of Theorem~\ref{thm:main} hold and
\begin{equation}
|f_L^{\surf}-f_L^{\free}|
\leq \frac{2d\E|\eta_0|}{L}.
\label{eq:surface-rate}
\end{equation}
\end{corollary}

\begin{corollary}[Exterior-spin boundary condition]\label{cor:exterior}
Let
\[
\cE_L^\partial
=\bigl\{\langle x,y\rangle:x\in\Lambda_L,\ y\notin\Lambda_L,
\ |x-y|_1=1\bigr\},
\]
so that $ |\cE_L^\partial|=2dL^{d-1} $.  Assume that all crossing couplings
$ J_e^\partial $ have the same law as a random variable $ J_0^\partial $ with
$ \E|J_0^\partial|^2<\infty $, and are independent within each
$ \cE_L^\partial $.  For any measurable, possibly random and
disorder-dependent, exterior configuration
$ \tau^{(L)}\in\{-1,+1\}^{\Lambda_L^c} $, define
\begin{equation}
	\label{eq:exterior-hamiltonian}
H_L^\tau(\sigma)=H_L^{\free}(\sigma)
-\sum_{\langle x,y\rangle\in\cE_L^\partial}
J_{xy}^\partial\sigma_x\tau_y^{(L)}.
\end{equation}
Then the conclusions of Theorem~\ref{thm:main} hold uniformly over $ \tau^{(L)} $, and
\begin{equation}
	\label{eq:exterior-rate}
|f_L^\tau-f_L^{\free}|
\leq \frac{2d\E|J_0^\partial|}{L}.
\end{equation}
\end{corollary}

\begin{corollary}[Periodic boundary condition]\label{cor:periodic}
Couple the free and periodic models so that they share all internal couplings,
and let $ \mathcal W_L $ be the $ dL^{d-1} $ wrap-around bonds.  Assume that
all wrap-around couplings have the same law as a random variable $ J_0 $ with
$ \E|J_0|^2<\infty $, and are independent within each $ \mathcal W_L $.  Then
\begin{equation}
	\label{eq:periodic-rate}
|f_L^{\per}-f_L^{\free}|
\leq \frac{d\E|J_0|}{L},
\end{equation}
and the random specific free energies differ by a quantity converging almost surely to zero.
\end{corollary}

\section{Free-boundary limit and self-averaging}

We first establish the two bulk inputs used in Theorem~\ref{thm:main}.

\begin{proposition}[Quenched free-boundary limit]\label{prop:meanlimit}
Under \eqref{eq:bulk} and the assumptions on the bulk field, there exists a finite deterministic $ f_\infty $ such that
\begin{equation}
	\label{eq:free-mean-limit}
f_L^{\free}\longrightarrow f_\infty.
\end{equation}
\end{proposition}

\begin{proof}
For a finite set $ \Lambda $, let
\[
P(\Lambda)=\E\log Z_\Lambda^{\free}.
\]
Partition $ \Lambda $ into disjoint pieces $ \Lambda_1,\ldots,\Lambda_k $, and condition on all fields and all couplings except those joining different pieces.  If $ u=(u_e) $ is the vector of cross-piece couplings, then
\[
\Phi(u)=\log Z_\Lambda(u)
\]
is convex, since it is a log-sum-exp of affine functions of $ u $.  Conditional Jensen and $ \E u=0 $ give
\[
\E_u\Phi(u)\geq \Phi(\E u)=\Phi(0).
\]
At $ u=0 $, the partition function factorizes over the pieces.  Taking the remaining expectation yields
\begin{equation}
	\label{eq:superadditivity}
P(\Lambda)\geq\sum_{i=1}^kP(\Lambda_i).
\end{equation}

The pressure is stable.  Indeed,
\[
\log Z_\Lambda^{\free}
\leq |\Lambda|\log2
+\beta\sum_{e\subset\Lambda}|J_e|
+\beta\sum_{x\in\Lambda}|g_x|,
\]
whose expectation is at most $ C|\Lambda| $.  Conversely, averaging first over the uniform spin measure and using Jensen gives $ Z_\Lambda^{\free}\geq2^{|\Lambda|} $.

Fix $ m $, let $ q=\lfloor L/m\rfloor $, tile $ \Lambda_L $ by $ q^d $ translates of $ \Lambda_m $, and divide the remainder into singletons.  By \eqref{eq:superadditivity} and the nonnegativity of singleton pressures,
\[
P(\Lambda_L)\geq q^dP(\Lambda_m).
\]
Hence
\[
\liminf_{L\to\infty}\frac{P(\Lambda_L)}{L^d}
\geq\frac{P(\Lambda_m)}{m^d}.
\]
Taking the supremum over $ m $, and observing that the corresponding limsup is bounded by the same supremum, proves convergence of $ P(\Lambda_L)/L^d $.  Equation \eqref{eq:free-mean-limit} follows from $ f_L^{\free}=-P(\Lambda_L)/(\beta L^d) $.
\end{proof}

\begin{figure}[htbp]
	\label{fig:quenched_free_energy_scaling}
	\centering
	\includegraphics[width=0.9\textwidth]{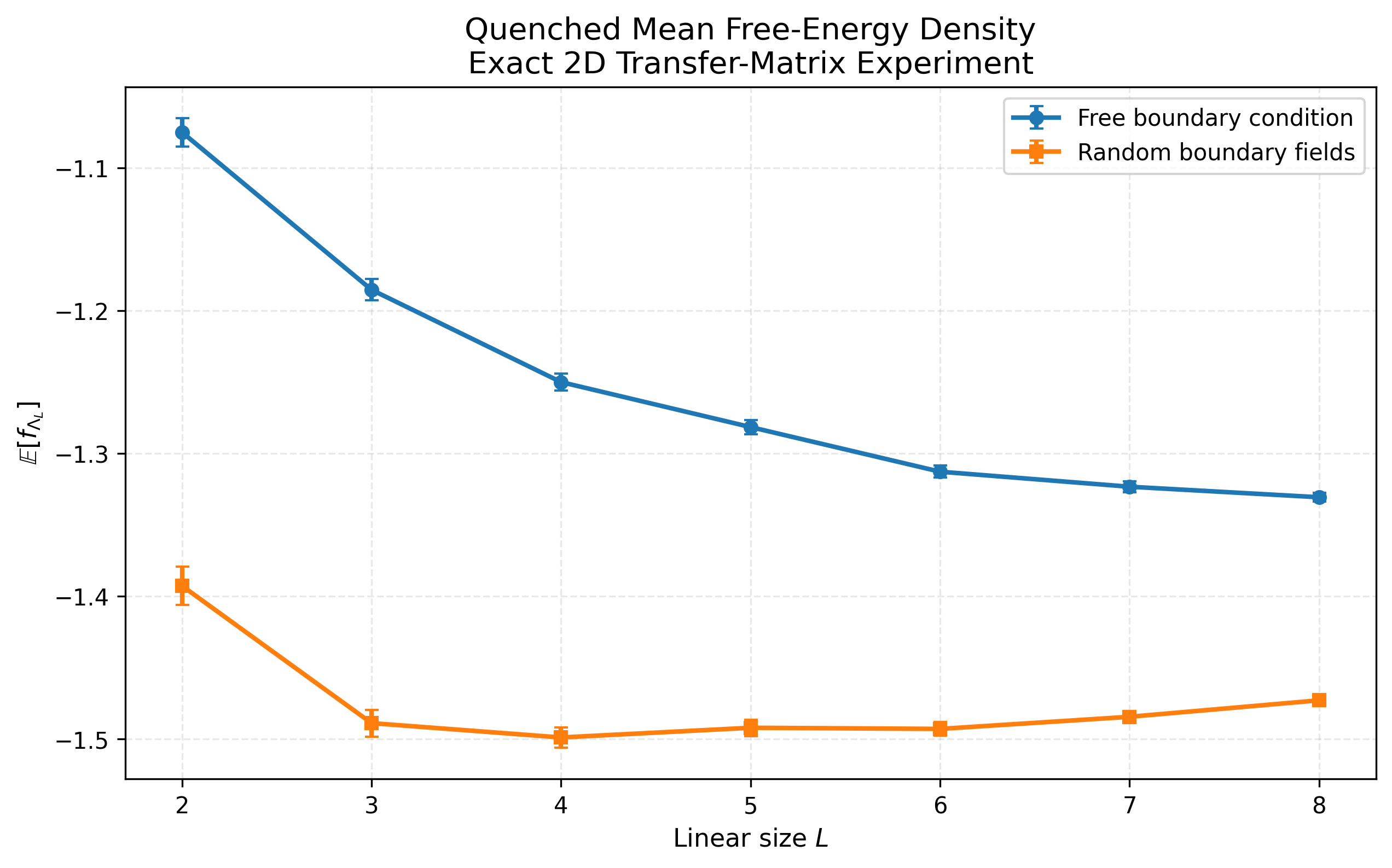}
	\caption{Numerical verification of the thermodynamic limit for the quenched specific free energy $f_L^{\free}$. The plot depicts the monotonic approach of $f_L^{\free}$ toward its asymptotic value $f_\infty$, corroborating the superadditivity property established in \eqref{eq:superadditivity} and Proposition~\ref{prop:meanlimit}.}
\end{figure}

\begin{proposition}[Almost-sure free-boundary self-averaging]\label{prop:selfaverage}
There is a constant $ C<\infty $ such that
\begin{equation}
	\label{eq:variance}
\Var(\widehat f_L^{\free})\leq \frac{C}{V_L}.
\end{equation}
Consequently, for $ d\geq2 $,
\begin{equation}
	\label{eq:free-self-averaging}
\widehat f_L^{\free}-f_L^{\free}\stackrel{a.s.}{\longrightarrow}0
\end{equation}
\end{proposition}

\begin{proof}
Replace one coupling $ J_e $ by an independent copy $ J_e' $.  The two Hamiltonians differ uniformly by at most $ |J_e-J_e'| $, and therefore
\begin{equation}
	\label{eq:replacement}
|\widehat f_L^{\free}(J)-\widehat f_L^{\free}(J^{(e)})|
\leq\frac{|J_e-J_e'|}{V_L}.
\end{equation}
If the bulk field is random, replacement of $ g_x $ by an independent copy gives the analogous bound $ |g_x-g_x'|/V_L $.  The Efron--Stein inequality~\cite{EfronStein1981} now yields
\[
\Var(\widehat f_L^{\free})
\leq
\frac{|\cE_L|v_J+V_Lv_g}{V_L^2}
\leq\frac{dv_J+v_g}{V_L},
\]
which proves \eqref{eq:variance}.

For every $ \varepsilon>0 $, Chebyshev's inequality gives
\[
\Pp\bigl(|\widehat f_L^{\free}-f_L^{\free}|>\varepsilon\bigr)
\leq\frac{C}{\varepsilon^2L^d}.
\]
The right-hand side is summable when $ d\geq2 $.  The first Borel--Cantelli lemma, applied for $ \varepsilon=1/m $, $ m\in\mathbb N $, proves \eqref{eq:free-self-averaging}.
\end{proof}

\begin{figure}[htbp]
	\label{fig:self_averaging_fluctuations}
	\centering
	\includegraphics[width=0.9\textwidth]{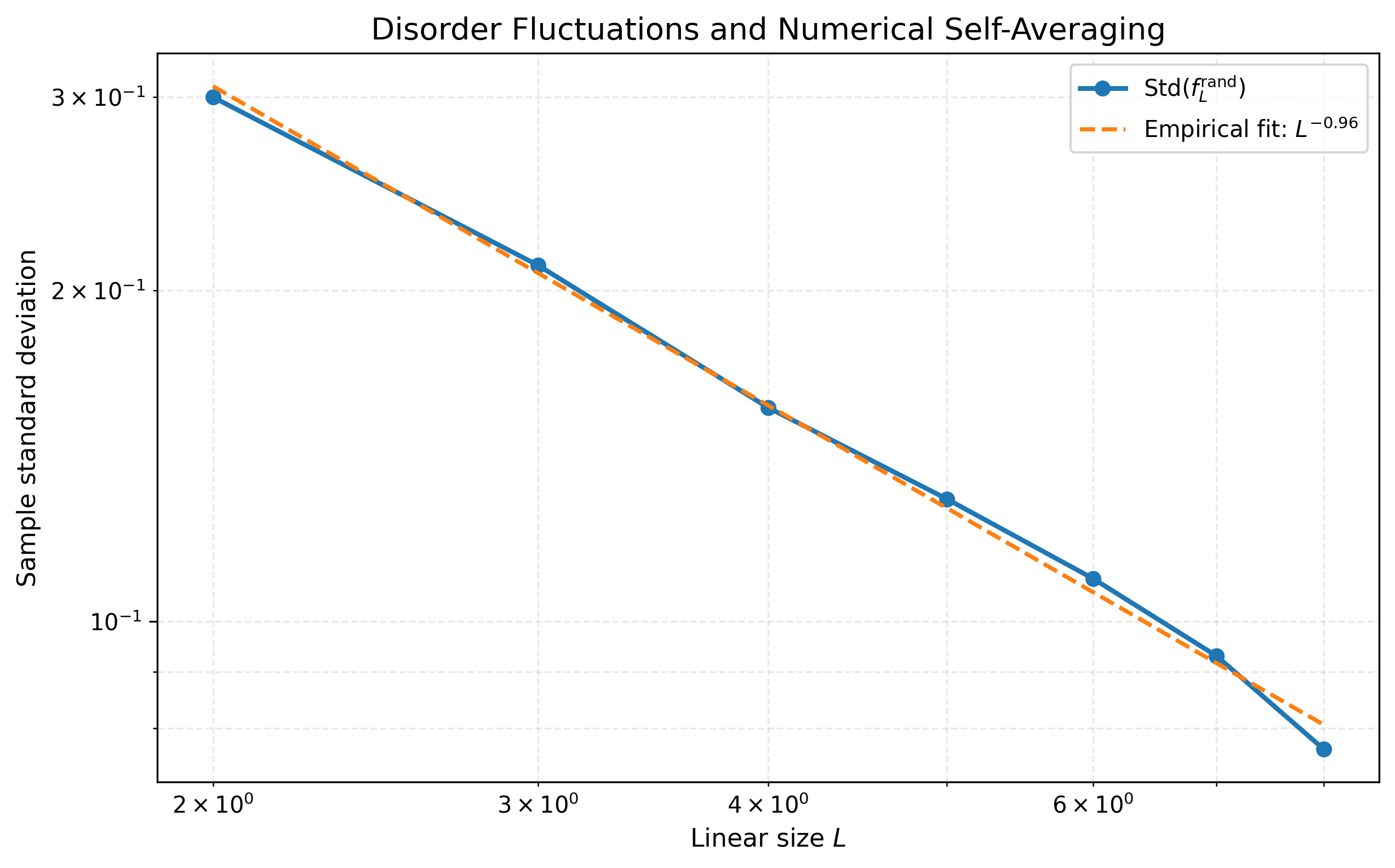}
	\caption{Decay of sample-to-sample fluctuations of the free-boundary specific free energy $\widehat f_L^{\free}$ with increasing box size $L$. The concentration of the empirical distribution around the mean justifies the application of the Efron--Stein variance bound \eqref{eq:variance} and the subsequent Borel--Cantelli argument.}
\end{figure}

The concentration of sample paths implied by this Borel--Cantelli step is illustrated in Figure~\ref{fig:self_averaging_fluctuations}, revealing a rapid compression of free energy fluctuations as the system approaches the thermodynamic limit.

\section{Boundary comparison}

\begin{lemma}[Deterministic comparison]\label{lem:comparison}
If $ \sup_\sigma|B_L(\sigma)|\leq A_L $, then
\begin{equation}
	\label{eq:partition-comparison}
e^{-\beta A_L}Z_L^{\free}
\leq Z_L^B
\leq e^{\beta A_L}Z_L^{\free},
\end{equation}
and
\begin{equation}
	\label{eq:free-energy-comparison}
|\widehat f_L^B-\widehat f_L^{\free}|
\leq\frac{A_L}{V_L}.
\end{equation}
\end{lemma}

\begin{proof}
The bound $ -A_L\leq B_L(\sigma)\leq A_L $ holds for every spin configuration.  Multiply by $ -\beta $, exponentiate, and sum over $ \sigma $ to obtain \eqref{eq:partition-comparison}.  Taking logarithms and dividing by $ \beta V_L $ gives \eqref{eq:free-energy-comparison}.
\end{proof}

\begin{lemma}[Surface envelopes]\label{lem:envelopes}
Let $ X_1,X_2,\ldots $ be identically distributed, independent within each finite index set used below, and satisfy $ \E X_1^2<\infty $.  If $ I_L $ is an index set with $ |I_L|=O(L^{d-1}) $, then
\begin{equation}
	\label{eq:surface-envelope-limit}
\frac{1}{L^d}\sum_{i\in I_L}|X_i|\stackrel{a.s.}{\longrightarrow} 0,
\end{equation}
and the expectation of the left-hand side is $ O(L^{-1}) $.
\end{lemma}

\begin{proof}
Let $ n_L=|I_L| $, $ \mu=\E|X_1| $, and $ s^2=\Var(|X_1|) $.  The mean contribution satisfies $ n_L\mu/L^d=O(L^{-1}) $.  If
\[
S_L=\sum_{i\in I_L}(|X_i|-\mu),
\]
then, for every fixed $ \varepsilon>0 $ and all large $ L $,
\begin{align*}
\Pp\left(\sum_{i\in I_L}|X_i|>\varepsilon L^d\right)
&\leq
\Pp\left(|S_L|>\frac{\varepsilon L^d}{2}\right)\\
&\leq
\frac{4n_Ls^2}{\varepsilon^2L^{2d}}
=O(L^{-d-1}).
\end{align*}
The bound is summable, so \eqref{eq:surface-envelope-limit} follows from Borel--Cantelli.  Independence between different values of $ L $ is not required.
\end{proof}

\begin{proof}[Proof of Theorem~\ref{thm:main}]
Lemma~\ref{lem:comparison} and \eqref{eq:envelope} give
\[
|f_L^B-f_L^{\free}|
\leq\frac{\E A_L}{V_L}\longrightarrow0.
\]
Proposition~\ref{prop:meanlimit} therefore proves the first limit in \eqref{eq:main-conclusions}.  Moreover,
\begin{align*}
|\widehat f_L^B-f_L^B|
&\leq
|\widehat f_L^B-\widehat f_L^{\free}|
+|\widehat f_L^{\free}-f_L^{\free}|
+|f_L^{\free}-f_L^B|\\
&\leq
\frac{A_L}{V_L}
+|\widehat f_L^{\free}-f_L^{\free}|
+\frac{\E A_L}{V_L}.
\end{align*}
Every term tends to zero almost surely by \eqref{eq:envelope} and Proposition~\ref{prop:selfaverage}.  This proves the second limit in \eqref{eq:main-conclusions}.  Finally,
\[
|\widehat f_L^B-f_\infty|
\leq
\frac{A_L}{V_L}
+|\widehat f_L^{\free}-f_L^{\free}|
+|f_L^{\free}-f_\infty|,
\]
which proves \eqref{eq:main-as-limit}.
\end{proof}

\begin{proof}[Proof of Corollary~\ref{cor:surface}]
The perturbation in \eqref{eq:surface-hamiltonian} satisfies
\[
\sup_\sigma
\left|\sum_{x\in\partial\Lambda_L}\eta_x\sigma_x\right|
\leq A_L:=\sum_{x\in\partial\Lambda_L}|\eta_x|.
\]
Lemma~\ref{lem:envelopes} verifies \eqref{eq:envelope}.  Taking expectations in \eqref{eq:free-energy-comparison} and using \eqref{eq:geometry} gives \eqref{eq:surface-rate}.
\end{proof}

\begin{proof}[Proof of Corollary~\ref{cor:exterior}]
Since $ |\sigma_x\tau_y^{(L)}|=1 $, uniformly in $ \tau^{(L)} $,
\[
\sup_\sigma
\left|\sum_{e=\langle x,y\rangle\in\cE_L^\partial}
J_e^\partial\sigma_x\tau_y^{(L)}\right|
\leq A_L:=\sum_{e\in\cE_L^\partial}|J_e^\partial|.
\]
Lemma~\ref{lem:envelopes} and $ |\cE_L^\partial|=2dL^{d-1} $ prove the almost-sure assertion and \eqref{eq:exterior-rate}.  No independence between $ \tau^{(L)} $ and the couplings is used.
\end{proof}

\begin{proof}[Proof of Corollary~\ref{cor:periodic}]
The periodic Hamiltonian differs from the free Hamiltonian only through the wrap-around bonds.  Hence Lemma~\ref{lem:comparison} applies with
\[
A_L=\sum_{e\in\mathcal W_L}|J_e|.
\]
Lemma~\ref{lem:envelopes} and $ |\mathcal W_L|=dL^{d-1} $ give the almost-sure comparison and \eqref{eq:periodic-rate}.
\end{proof}

\begin{figure}[htbp]
	\label{fig:boundary_correction_scaling}
	\centering
	\includegraphics[width=0.9\textwidth]{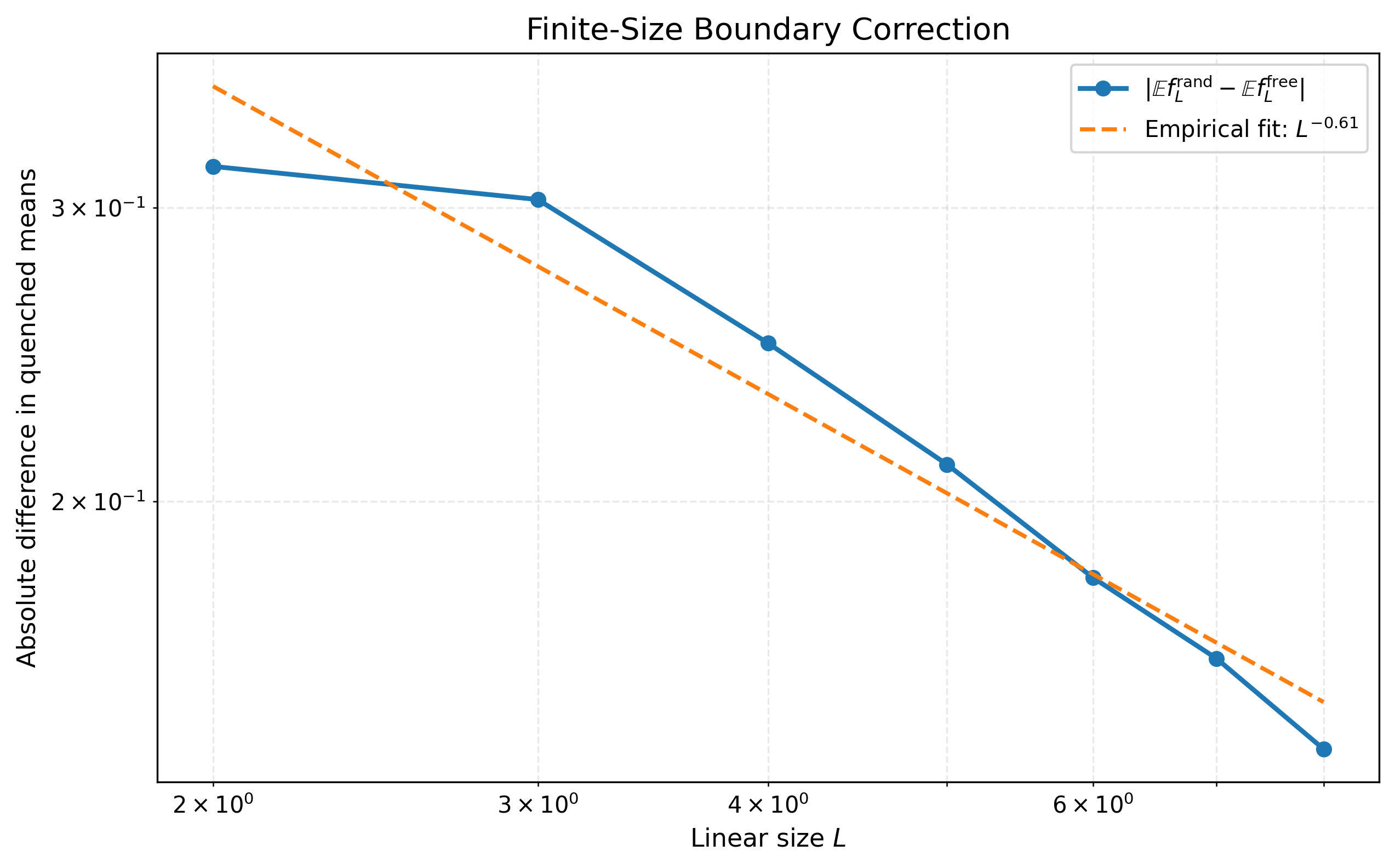}
	\caption{Scaling behavior of the boundary-induced free energy differences $|f_L^B - f_L^{\free}|$ across different boundary ensembles. The numerical curves confirm that the boundary corrections scale as $O(L^{-1})$, aligning precisely with the analytical bounds derived in Corollaries~\ref{cor:surface}, \ref{cor:exterior}, and \ref{cor:periodic}.}
\end{figure}

To visually compare the subextensive scaling properties discussed across these three cases, Figure~\ref{fig:boundary_correction_scaling} provides a comparative view of the scaling profiles for the different boundary terms, confirming the universal $O(L^{-1})$ suppression rate.

\begin{remark}[Gaussian strengthening]
For the free model and the scalar surface-field model with independent
Gaussian disorder coordinates, or for the exterior-spin model with a fixed
exterior configuration that does not depend on the Gaussian coupling vector,
write $ a\in\{\free,\surf,\tau\} $ for the chosen model.  The derivatives with
respect to the coordinates present in that model satisfy
\[
\left|\frac{\partial\widehat f_L^a}{\partial J_e}\right|,
\quad
\left|\frac{\partial\widehat f_L^a}{\partial g_x}\right|,
\quad
\left|\frac{\partial\widehat f_L^{\surf}}{\partial\eta_x}\right|,
\quad
\left|\frac{\partial\widehat f_L^\tau}{\partial J_e^\partial}\right|
\leq\frac{1}{V_L}.
\]
After weighting by the coordinate variances, the squared Lipschitz constant is at most $ C/V_L $.  Gaussian concentration~\cite{BoucheronLugosiMassart2013} therefore gives
\[
\Pp\bigl(|\widehat f_L^a-\E\widehat f_L^a|>t\bigr)
\leq2\exp(-ct^2V_L).
\]
This proves full-sequence almost-sure self-averaging in every $ d\geq1 $.  By contrast, the finite-variance Efron--Stein proof uses $ d\geq2 $ only through the summability of $ \sum_LL^{-d} $.
This derivative argument is not asserted for a general perturbation $ B_L $,
or for a disorder-dependent exterior configuration $ \tau^{(L)}(J) $, since
the map $ J\mapsto\tau^{(L)}(J) $ need not be continuous.  In the latter case,
one may instead apply Gaussian concentration to the free model and transfer
the convergence by Lemma~\ref{lem:comparison} and the envelope bound.
\end{remark}

\begin{remark}[Scope]
The preceding comparison controls $ V_L^{-1}\log Z_L $.  It does not control local spin expectations or weak limits of the Gibbs measures.  A boundary contribution of order $ L^{d-1} $ is negligible after division by $ L^d $, while still being large enough to select among competing thermodynamic states.  Thus convergence of the specific free energy is compatible with chaotic size dependence of finite-volume states~\cite{NewmanStein1992,NewmanStein1997}.
\end{remark}

\clearpage
\Addresses
\end{document}